\documentclass[runningheads,a4paper]{llncs}
\usepackage{amsfonts,amssymb,latexsym,xspace}
\usepackage{amsmath,enumerate}
\usepackage[english]{babel}
\usepackage{graphicx}
\usepackage{array,multirow,makecell}

 \def\cal{\mathcal}
\def\E{\mathbb E}
\def\hat{\widehat}
\def \und{\underline}
\def \Q {{\mathbb Q}_2}
\def \a {{\boldsymbol a}}

\begin{document}
\title{Analysis of the Continued Logarithm Algorithm}
\author{Pablo Rotondo\inst{1,2,3} \and  Brigitte Vall\'ee\inst{2} \and Alfredo Viola\inst{3}}
\institute{IRIF, {CNRS and  Universit\'e  Paris Diderot}, France
\and 	
GREYC,  {CNRS and Universit\'e de Caen},   France
\and
Universidad de la Rep\'ublica, Uruguay}

\maketitle

\begin{abstract}
The Continued Logarithm Algorithm  --{\em CL }for short-- introduced by Gosper in 1978
computes the gcd of two integers; it seems very efficient, as it only performs shifts and subtractions.
Shallit has studied its worst-case complexity in 2016 and showed it to be linear.
We here perform the average-case analysis of the algorithm : we study its  main parameters (number of iterations,  total number of shifts)
  and  obtain  precise  asymptotics for their mean values.
Our  ``dynamical'' analysis involves the  dynamical system underlying the algorithm, that  produces continued fraction expansions    whose quotients are  powers of 2.  Even though this  {\em CL} system  has already been   studied  by Chan (around 2005), 
  the  presence of powers of 2  in the quotients  ingrains into the central parameters
a dyadic flavour that cannot be grasped  solely by studying the {\em CL} system. We  thus introduce a dyadic component and deal with a two-component system.  With this new  mixed system at hand,  we then 
provide  a complete average-case analysis of the {\em CL} algorithm, with explicit constants\footnote{\small {Accepted in Latin American Theoretical Informatics 2018. \\Thanks to the {\tt Dyna3S} ANR Project 
and the {\tt AleaEnAmsud} AmSud-STIC Project.}}.
\end{abstract}

\section{Introduction}

In an unpublished manuscript \cite{gosper}, Gosper introduced the continued logarithms, a mutation of
the classical continued fractions. He writes
{``The primary advantage is the conveniently small information parcel. The restriction to integers
of regular continued fractions makes them unsuitable for very large and very small numbers. The
continued fraction for Avogadro's number, for example, cannot even be determined to one term, since
its integer part contains 23 digits, only 6 of which are known. (...) By contrast, the continued logarithm of
Avogadro's number begins with its binary order of magnitude,  and only then begins the description
equivalent to the leading digits -- a sort of recursive version of scientific notation''.}

The  idea of Gosper gives rise to an algorithm for computing  gcd's, described by Shallit in \cite{shallit}. This algorithm has  two advantages:  first,  it can be calculated starting from the most representative  bits, and uses very simple operations (subtractions and shifts); it does not employ divisions. Second,  as the  quotients which intervene in  the  associated continued fraction  are powers of two  $2^a$, we can store each  of them with $\log_2 a $ bits. Then, the algorithm seems to be of small complexity,   both in terms of computation and storage. 

\smallskip
Shallit \cite{shallit}	  performs the worst-case analysis of the algorithm,  and studies the number of steps $K(p, q)$, and the total number of shifts $S(p, q)$ that are performed on an integer input $(p, q)$ with $p < q$: he proves the inequalities
\begin{equation}  \label{worstcase}
 K(p, q) \le 2\log_2  q +2, \qquad  S(p, q)  \le (2 \log_2 q +2)\,\log_2 q \, , 
 \end{equation}   
and exhibits instances, namely the family  $(2^{n-1}, 1)$,  which show that the previous bounds are nearly optimal, 
$ K( 1, 2^{n-1}) = 2n-2, \ \   S(1,2^{n-1}) = n(n-1)/ 2 +1\, .$

\smallskip
In a personal communication, Shallit proposed us to perform the average-case analysis of the algorithm. In this paper, we answer his question. We consider the set   $\Omega_N$  which gathers the integer pairs $(p, q)$  with $0\le p \le q \le N$, endowed with the uniform probability, and we study the  mean values $\E_N[K]$ and $\E_N[ S]$ as $N \to \infty$. We prove  that these mean values are  asymptotically linear in the size $\log N$,  and exhibit their precise asymptotic behaviour for $N \to \infty$,  

\vspace{-0.25cm} 
$$ \E_N[K]   \sim  \frac 2{H} \log N , \quad E_N[S] \sim  \frac { \log 3 -\log 2}{2 \log 2\! -\!\log 3} \ �\E_N[K]  \, . $$
 The constant $H$ is related to the entropy of an associated dynamical system and 
 
 \vspace{-0.25cm}
 \begin{equation} \label{const}
 H= \frac 1 {2\log 2 \!-\! \log 3 } \!\!\left[  \frac {\pi^2}{6} +  2 \sum_{k \ge 1}  \frac {(-1)^k}  {k^2\,  2^k} - (\log 2 )(3 \log 3\!- \!4 \log 2) \right] .
 \end{equation} This entails  numerical estimates (validated by experiments) for the  mean values
 \vskip 0.1cm
 \centerline {  $\E_N[K]  \sim  1.49283  \log N, \qquad   \E_N[S] \sim 1.40942 \log N$.}
 \vskip 0.1cm 
Then, from \eqref{worstcase},  the  mean number of divisions is about half the maximum. 

\smallskip
 Our initial idea was  to perform a dynamical analysis along the lines described in \cite {Va2}. The  {\em CL} algorithm is defined as  a succession of steps, each consisting of a pseudo-division which  transforms  an   integer pair  
 into a new one. 
 This transformation may be read first on the associated rationals 
  and gives rise to  a mapping $T$ that is further extended to the real unit interval ${\cal I}$. 
  This smoothly yields a dynamical  system  $({\cal I}, T)$, the {\em CL} system,  already well studied, particularly by Chan \cite{chan1} and Borwein \cite{bor}. The system has an invariant density  $\psi$ (with an explicit expression described in   \eqref{psi}) and is ergodic. Thus we  expected  this dynamic analysis to  follow general  principles   described in \cite{Va2}.

 \medskip
 However, the analysis  of the algorithm is not so straightforward. The binary shifts, which  make the algorithm so efficient, cause many problems in the analysis.  Even on a  pair of coprime integers $(p, q)$, the algorithm  creates  intermediate pairs  $(q_{i+1}, q_i) $ which are no longer generally coprime, as their gcd is a non-trivial  power of 2. These extra gcd's are  central in the analysis of the algorithm, as they have an influence on the evolution of the sizes of the pairs  $(q_{i +1}, q_i)$ which may grow due these extra factors. As these extra factors may only be powers of 2, they are easily expressed  with the dyadic absolute value on ${\mathbb Q}$,  at least when the input 
 is rational. However, when  extending the algorithm into a  dynamical system on the unit interval, we lose  track of these factors. 
 
 \medskip 
 The (natural)  idea is thus  to  add to the usual {\em CL} dynamical system (on the unit interval ${\cal I}$)  a new component in the dyadic  field  $\Q$. 
 The dyadic component is just added  here to deal with the extra dyadic factors, as a sort of accumulator, but it is the former real component that dictates the evolution of the system.  As the initial {\em CL} system has nice properties, the mixed system inherits this good behaviour. In particular,   the  transfer operator  of the mixed system presents a dominant eigenvalue, and  the dynamical analysis may be performed successfully.  The constant $H$ in Eqn \eqref{const}  is actually  the entropy of this extended system.

 \medskip
 After this extension, the analysis follows  classical steps,  with methodology mixing tools from analytic combinatorics (generating functions, here of Dirichlet type),  Tauberian theorems (relating the singularities of these generating functions to the asymptotics of their coefficients), functional analysis (which transfers the geometry of the dynamical system into  spectral properties of the transfer operator). 
 
 \vspace{-0.2cm} 
 \paragraph{\bf \em Plan of the paper and notation.}  The paper is structured into three sections.  Section 2 introduces the algorithm and its associated dynamical system, as well as  the probabilistic model,  the costs of interest and their generating functions. Then, Section 3 defines  the extended dynamical system, allowing us to work with dyadic costs; it  explains how the corresponding transfer operator  provides alternative expressions for the generating functions. Finally,  Section 4 describes  the properties of  the transfer operator, namely its dominant spectral properties on a convenient functional space. With Tauberian theorem, it  provides the final  asymptotic estimates  for the mean values of the main costs of interest. 
 
\smallskip 
 For an integer $q$,  $\delta(q)$ denotes the dyadic valuation, i.e.,   is the greatest integer $k$  for which $2^k$  divides $q$.  The dyadic norm  $|\cdot|_2$ is defined on ${\mathbb Q}$   with the equality  $|a/b|_2 := 2^{\delta(b)-\delta(a)}$. The  dyadic field  $\mathbb{Q}_2$ is the completion of ${\mathbb Q}$ for this norm. See \cite{Ko} for more details about the dyadic field $\Q$.

\section{The {\em CL} algorithm and  its dynamical system.}

We kick off
 this section with a precise description of the {\em CL} algorithm, followed by its extension to the whole of the unit interval $\cal I$, giving rise to a dynamical system, called the {\em CL} system, whose inverse branches  capture  all of our costs of interest.  
 Then 
 we recall the already known features   of the {\em CL} system and we present the  probabilistic model, with   its generating functions.

\vspace{-0.2cm}
\paragraph{\bf \em Description of the algorithm.} The algorithm, described by Shallit in \cite{shallit},  is a sequence of (pseudo)--divisions: each division   associates to  a pair    $(p, q)$\footnote{Our notations are not the same as in the paper of Shallit as we reverse the roles of $p$ and $q$.} with $p<q$   a new pair $(r, p')$ (where $r$ is the ``remainder'')  defined  as follows
$$ q = 2^{a} p+ r,\quad  p' = 2^{a} p, \qquad \hbox{ with} \quad  a=a(p,q):= \max\{k \ge 0 \mid 2^k p\leq q\}\, .$$ 
This transformation rewrites the old pair $(p, q)$ in terms of the new one
$(r, 2^a p)$   in matrix form, 
\begin{equation}\label{Ma} \left(\begin{matrix} p  \\ q\end{matrix}\right) =  N_{a}   \left(\begin{matrix} r \\2^a p \end{matrix} \right), \quad \hbox{with}\ \    
   N_a  =\left(\begin{matrix} 0 & \ \  2^{-a} \\
1  & 1 
\end{matrix}\right)  = 
   2^{-a}\,  M_a, \quad M_a = \left(\begin{matrix} 0 &  1 \\
2^{a}  & 2^a 
\end{matrix}\right) \, .
\end{equation}
The {\em CL} algorithm  begins   with  the input $(p, q)$ with $p <q$. It  lets $(q_1, q_0) := (p, q)$, then performs a sequence of divisions  

\vskip 0.1cm
\centerline{
$\left(\begin{matrix} q_{i +1}, \   q_i \end{matrix}\right)^T = N_{a_{i +1}}\     \left(\begin{matrix} q_{i +2}, 2 ^{a_{i +1}} \ q_{i +1}\end{matrix} \right)^T\,  { ,}$}
\vskip 0.1cm
  and stops  after $ k =  K(p, q)$ steps  on a pair  of the form $ (0, 2^{a_k} q_{k} )$.  
The complete execution of the algorithm   uses  the set of matrices $N_a$  defined in \eqref{Ma}, and writes the input as

\vskip 0.1cm
\centerline{
$ \left(\begin{matrix} p, \  q \end{matrix}\right)^T =  N_{a_1}\cdot N_{a_2} \cdots  N_{a_k}  \     \left(\begin{matrix} 0, \  2^{a_k} q_k \end{matrix} \right)^T\, { .}$}
\vskip 0.1cm
 The rational input $p/q$ is  then written as a continued fraction according to the LFTs  (linear fractional transformations) $h_a$ associated with matrices $N_a$  or  $M_a$, 
\begin{equation} \label{CF}
\frac{p}{q}  =
\cfrac{2^{-a_1}}{1+\cfrac{2^{-a_2}}{1+\cfrac{2^{-a_3}}{1+\ddots\cfrac{2^{-a_k}}{1}}}} = h_{a_1} \circ h_{a_2} \circ \cdots \circ h_{a_k} (0) 
\,, \quad  \hbox{with} \ \  h_a:x \mapsto \frac {2^{-a}} {1+x}\, .
\end{equation}

 Moreover, it is possible to choose the last exponent $a_k$ to be $0$.  (and the last quotient  to  be $1$). 
 This is a gcd algorithm:  as  $q_k$ is equal to $\gcd (p, q)$ up to a power of 2,   the {\em CL} algorithm determines the odd part of  $\gcd(p, q)$ whereas the even part is directly determined by the dyadic valuations of $p$ and $q$.

\smallskip
Shallit \cite{shallit}  proves that this algorithm indeed terminates and  characterizes the worst-case complexity of the algorithm.    Figure \ref{exe}  describes  the execution of the algorithm on the pair $(31, 75)$.   

 \setlength{\tabcolsep}{4pt}
 \setcellgapes{3pt}
 \makegapedcells

 \begin{figure}
 \begin{center} 
\begin{tabular}{|c|c|c|c|r|r|c|c|c|}
\hline
$i$ & $a_i$ & $2^{a_{i}} q_i$ & $q_{i+1}$ & $\left(2^{a_{i}} q_i\right)_2$ & $\left(q_{i+1}\right)_2$ &
$\delta(2^{a_{i}}q_{i})$ & $\delta(q_{i+1})$ & $\delta({\hat{g}_i })$ \\ \hline
$0$ & $-$ & $75$ & $31$ & 1001011 & 11111 & $0$
& $0$ & $0$ \\
$1$ & $1$ & $62$ & $13$ & 0111110 & 1101 & $1$
& $0$ & $0$ \\
$2$ & $2$ & $52$ & $10$ & 110100 & 1010 & $2$
& $1$ & $1$ \\
$3$ & $2$ & $40$ & $12$ & 101000 & 1100 & $3$
& $2$ & $2$ \\
$4$ & $1$ & $24$ & $16$ & 11000 & 10000 & $3$
& $4$ & $3$ \\
$5$ & $0$ & $16$ & $8$ & 10000 & 1000 & $4$
& $3$ & $3$ \\
$6$ & $0$ & $8$ & $8$ & 1000 & 1000 & $3$ &
$3$ & $3$ \\
$7$ & $-$ & $8$ & $0$ & 1000 & 0 & $3$ & $\infty$ & $3$ \\
\hline
\end{tabular}
\end{center}
\vskip -0.3cm

\caption {\label {exe}\small{Execution for the input $(31, 75)$.
 Here $\hat{g}_i=\gcd(2^{a_{i}} q_i,q_{i+1})$. The  dyadic valuation   $\delta(\hat{g}_i)$ seems to  linearly increase with $i$, with  $\delta(\hat{g}_i) \sim \delta(q_{i+1}) \sim i/2 $ ($i \to \infty$). }}
\end{figure}

\vspace {-0.4cm}
 \paragraph{\bf \em Dynamical system.} 
The relations 

\vskip 0.1 cm 
\centerline{
$ (r, 2^a p)^T = N_a^{-1} (p, q)^T, \qquad  (p, q)^T = N_a (r, 2^a p)^T\,,$}
\vskip 0.1cm 
are first transformed into relations on the associated rationals $p/q, r/(2^ap)$ via the LFT's  $T_a, h_a$  associated to  matrices $N_a^{-1}$, $N_a$, 
\begin{equation}
\label{eq:branchescl}
T_a(x) := \frac{1}{2^a x} - 1\,, \qquad h_a(x) = \frac{1}{2^a (1+x)}\,,\qquad a\geq 0\,.
\end{equation} They are then extended 
to the reals of the unit interval ${\cal I} := [0, 1]$. 
  This gives rise to a dynamical system $(\mathcal{I},T)$, denoted {\em CL} in the sequel, defined on the unit interval $\mathcal{I}$, with fundamental intervals $\mathcal{I}_a := [2^{-a-1},2^{-a}]$, the surjective branches $T_a\colon \mathcal{I}_a\to \mathcal{I}$, and their inverses $h_a\colon \mathcal{I} \to \mathcal{I}_a$.

\smallskip
 The {\em CL}  system  $({\cal I} , T)$ is displayed on the left of the figure below, along with the  shift $S: {\cal I} \rightarrow {\cal I}$ 
which gives rise to  the {\em CL} system by induction on the first branch. The map $S$ is   
 a mix of the Binary and Farey  maps, as its   first branch  comes from the Binary system, and the second one from the Farey system.  On the right,  the usual Euclid dynamical system (defined from  the Gauss map) is  derived from the Farey  shift by induction on the first branch.

\begin{center}
 \begin{tabular}{|c|c||c|c|}
\hline
  \includegraphics[width=.2\linewidth]{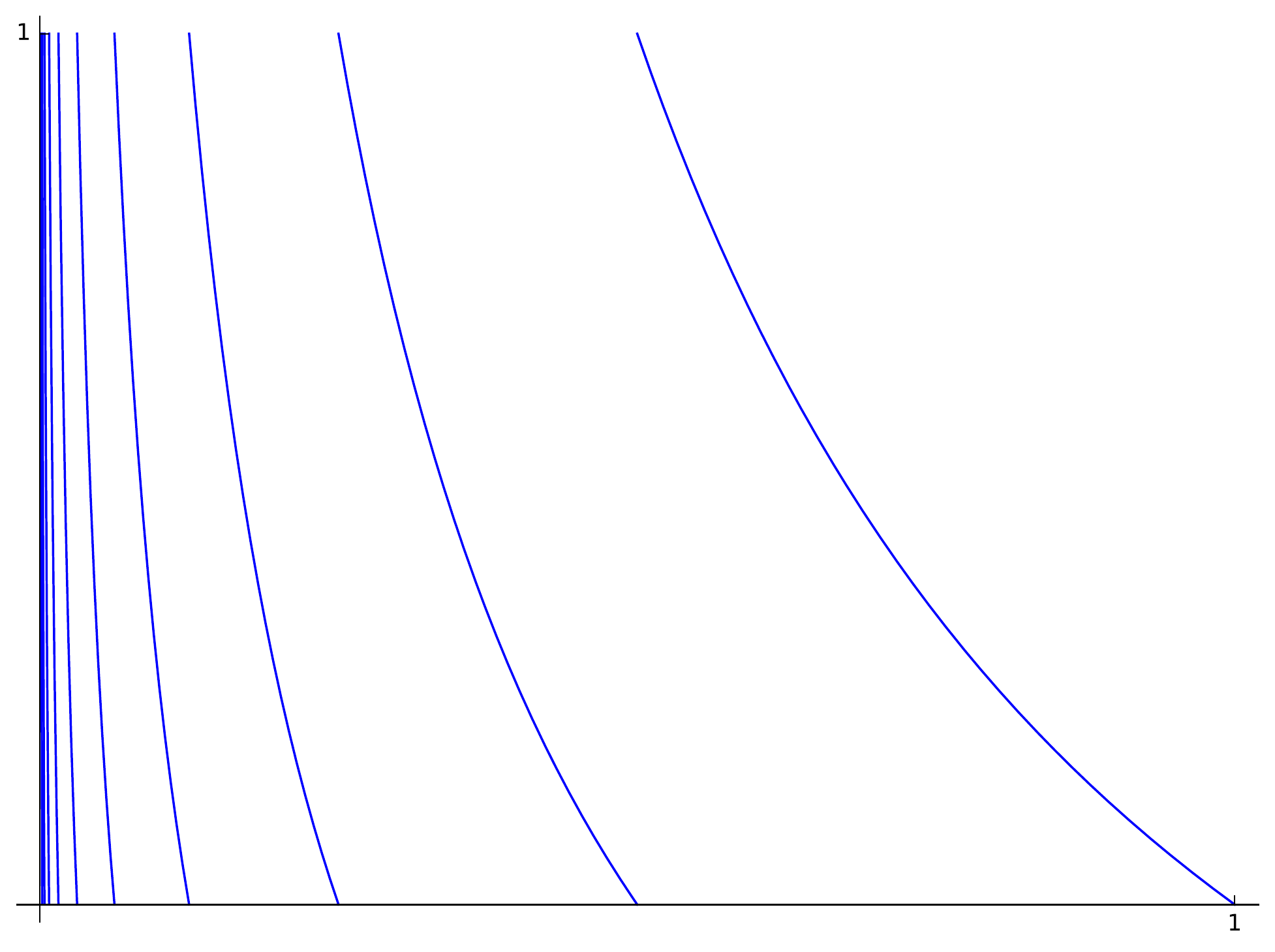}
&
  \includegraphics[width=.2\linewidth]{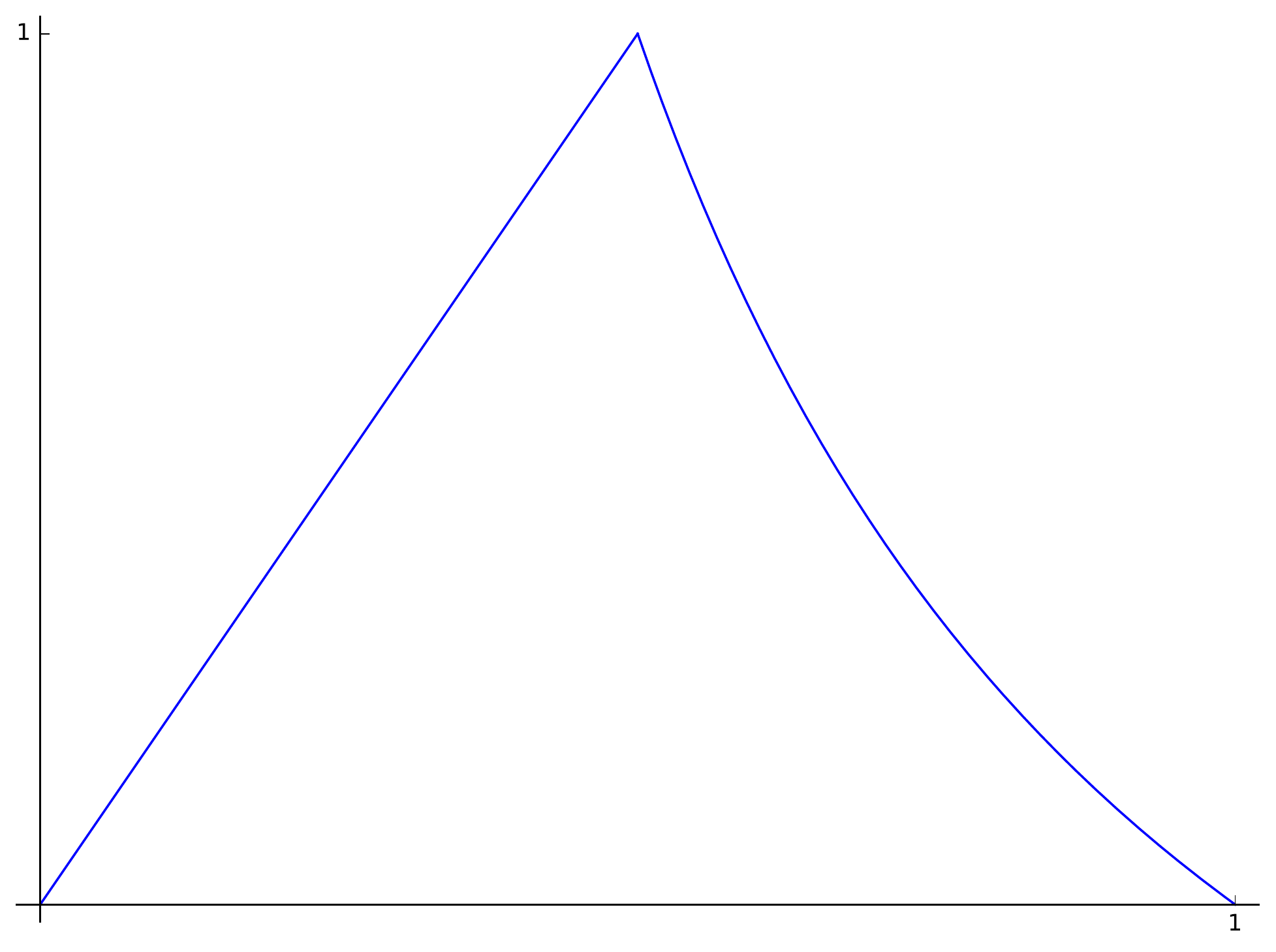} &
    \includegraphics[width=.2\linewidth]{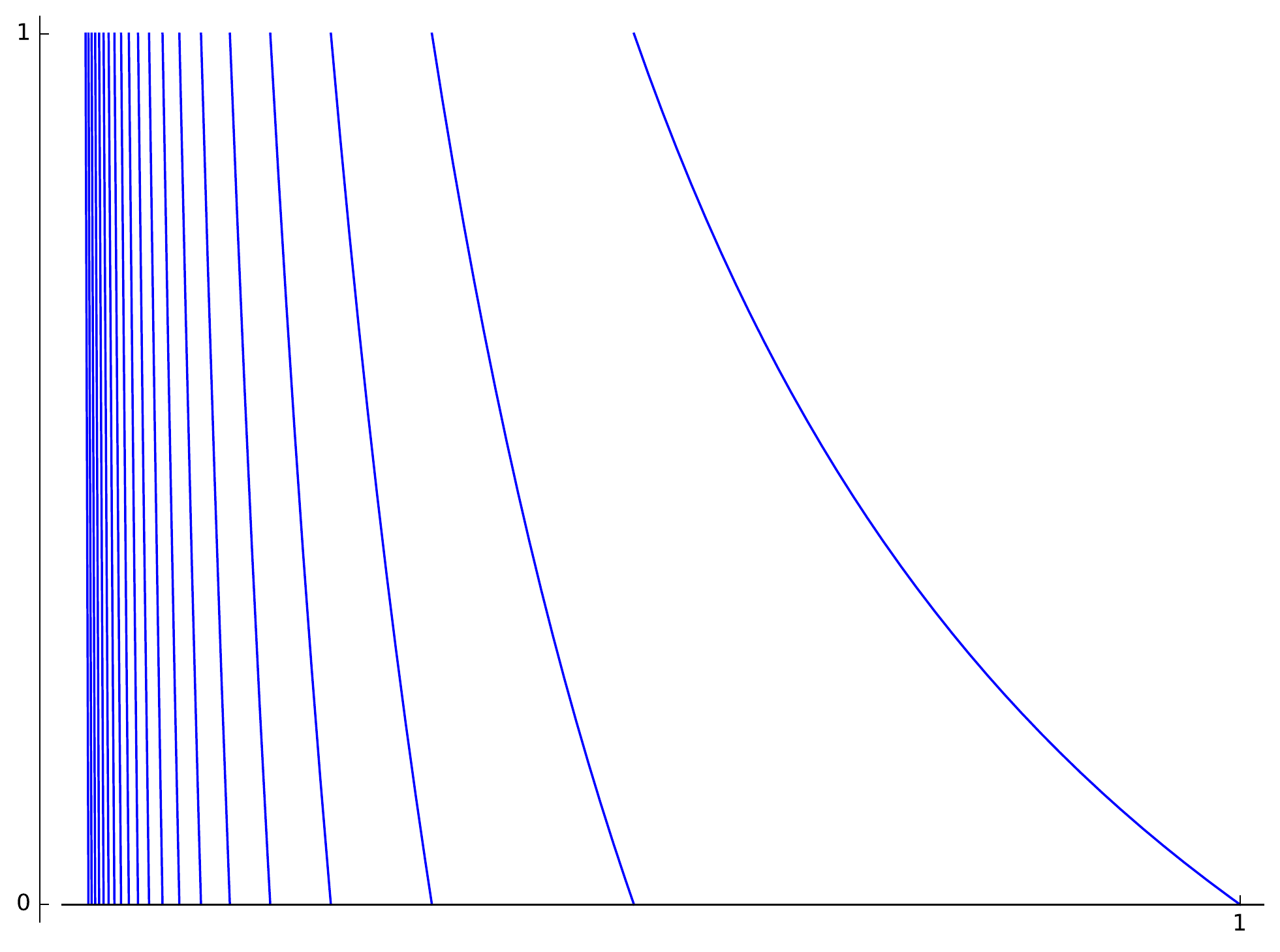}

&
  \includegraphics[width=.2\linewidth]{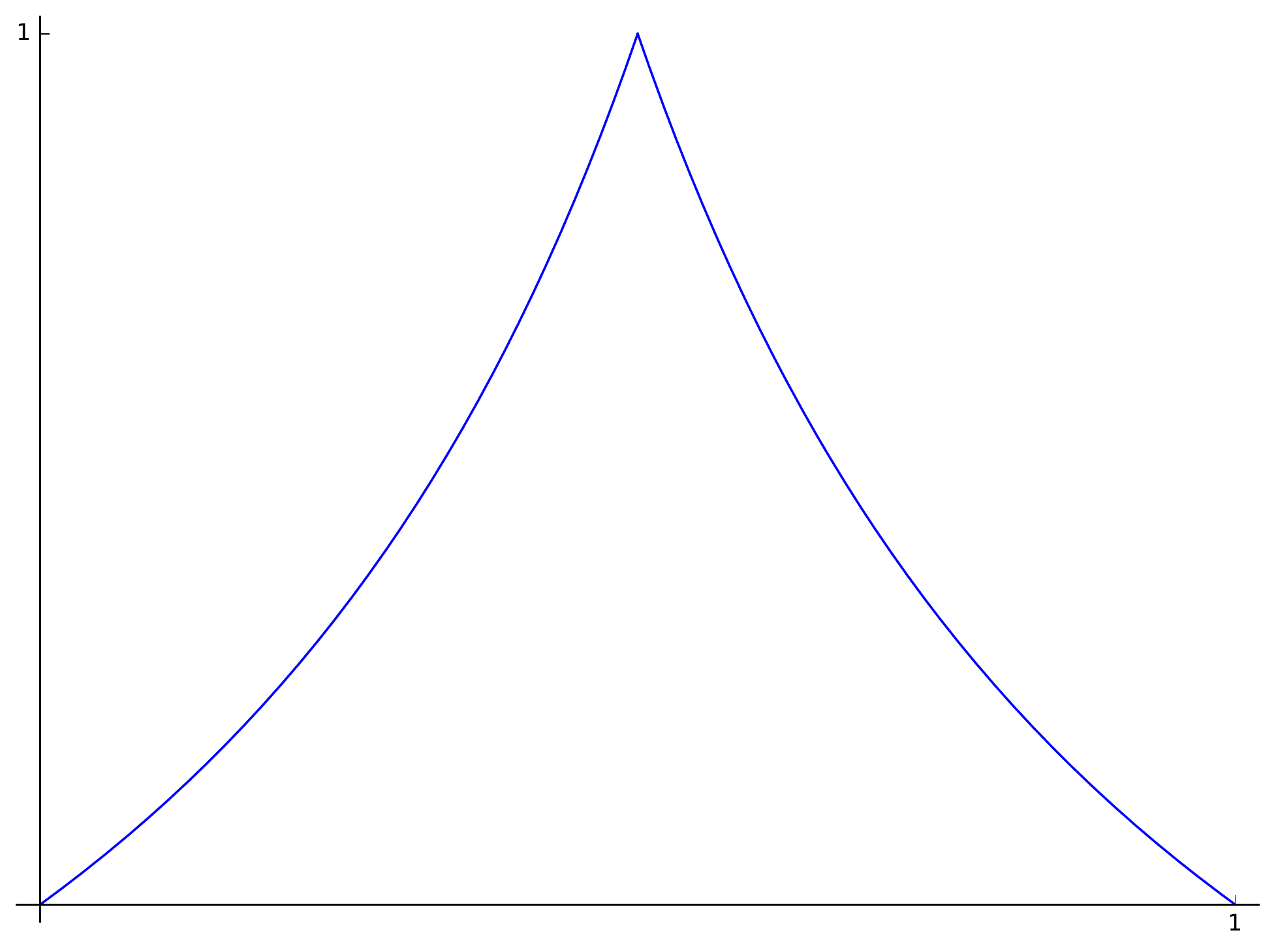} \\ \hline 
\end{tabular}
\end{center}

\smallskip
 With each $k$-uple $\a:=  \langle a_1,a_2,\ldots,a_k\rangle \in {\mathbb N}^k$  we associate  the matrix $M_\a := M_{a_1}\cdots M_{a_k}$ 
 and the LFT  $ h_\a := h_{a_1} \circ h_{a_2} \circ \cdots \circ h_{a_k} $. Then the set  ${\cal H}$  of the inverse branches, and the set  ${\cal H}^k$  of the inverse branches of $T^k$ (of depth $k$) are  
 \vskip 0.1cm 
 \centerline 
 { ${\cal H} := \{ h_a\mid a \ge 0 \} , \qquad {\cal H}^k := \{ h_\a \mid \a \in {\mathbb N}^k \}$. }
  \vskip 0.1 cm As the branches $T_a$ are surjective, the inverse branches are defined on ${\cal I}$ and the images $h_a({\cal I})$  for $a \ge 0$ form a topological partition of ${\cal I}$. This will be true at any depth, and the  intervals $h_\a({\cal I})$,  called fundamental intervals of depth $k$, form a topological partition  for $\a \in {\mathbb N}^k$.

   \vspace{-0.2cm}
   \paragraph{\bf \em  Properties of the {\em CL} system.} 
 The Perron Frobenius operator 
  \begin{equation} \label{H}
{\bf H}[f](x) := \sum_{h\in \mathcal{H}} |h^\prime(x)|\,  f(h(x)) =  \left(\frac 1 {1+x}\right)^2 \sum_{a \ge 0} 2^{-a} \,  f\left( \frac {2^{-a}}{1+x}\right)\,.
\end{equation}
  describes the evolution of densities : If $f$ is the initial density, ${\bf H}[f]$ is the density after one iteration of the system $({\cal I}, T)$.  The invariant  density $\psi$  is a fixed point  for ${\bf H}$ and satisfies the functional equation
  \begin{equation} \label {funceqn} 
   \psi(x) =  \left(\frac 1 {1+x}\right)^2 \ \sum_{a \ge 0} \,   2^{-a} \,  \psi\left( \frac {2^{-a}} {1+x}\right) \, .
   \end{equation}   Chan \cite{chan1} obtains an 
   explicit form for $\psi$
 \vspace{-0.2cm}
 \begin{equation}\label{psi} 
\psi(x) = \frac{1}{\log(4/3)} \frac{1}{(x+1)(x+2)}\, .  \end{equation}
  He  also  proves  that the system is  ergodic with respect  to $\psi$, and entropic. However, he does not provide an explicit expression for the  entropy. 
 We obtain here such an expression,  with a precise study of the transfer operator  of  the system. 
  
We introduce two (complex)  parameters $t, v$ in \eqref{H}, and deal with  a perturbation of the  operator ${\bf H}$, defined by
 \begin{equation} 
{\bf H}_{t, v} [f](x) := \!\!\sum_{h\in \mathcal{H}} |h^\prime(x)|^t \, d(h)^v \,  f(h(x)) = \!\! \left(\frac 1 {1+x}\right)^{\!2t} \!\sum_{a \ge 0} 2^{\,a(v-t)}  f\left( \frac {2^{-a}}{1+x}\right)\,.
\end{equation}
 Such an operator ${\bf H}_{t, v}$ is called a transfer operator. When $(t, v)$ satisfies $\Re (t-v)>0$, we prove  the following:  the operator ${\bf H}_{t, v}$ acts nicely on the space ${\cal C}^1({\cal I})$ 
  endowed with the norm $|\cdot |_{1, 1}$,  defined by
$ |f|_{1, 1}  :=  |f |_{0} + |f' |_{0}$, where $|\cdot|_0$ denotes the sup norm. In particular, it 
 has a dominant eigenvalue   $\lambda(t, v)$ separated from the remainder of the spectrum by a spectral gap, 
 for $(t, v)$  close to $(1, 0)$. The Taylor expansion of  $\lambda(t, v)$ near $(1, 0)$  \vskip 0.1cm 
\centerline{
$\lambda(t, v) \approx  1- A(t-1) + Dv $}
involves  the two constants  \ \ $  A = - {\partial\lambda}/{\partial t}  (1, 1, 0) ,   \ \   D =   {\partial \lambda}/{\partial v} (1, 1, 0)\,  ,  $  that 
  are expressed as mean values with respect to the invariant density $\psi$,  
\begin{equation} \label{A}
 A = E-D, \quad   E =  \E_\psi [2 |\log x|] , 
\quad D =  (\log 2) \,  \E_\psi[\und a]  \,,
\end{equation}
 (here,  the  function $\und a$  associates with  $x$ the integer  defined with the Iverson bracket $\und a(x) :=  a \cdot [\![x \in h_a({\cal I})]\!]$.     The constants   $A$ is the entropy of the system, and $E, D$ admit   explicit  expressions
  \vspace{-0.2cm}
  \begin{equation} \label{ED}
  E =    \frac 1 {\log (4/3) } \left[ \frac {\pi^2}{6} + 2 \sum_{k \ge 1} \frac {(-1)^k}  {k^2\,  2^k}\right], \quad  D = (\log 2) \,   \frac { \log (3/2)}{\log (4/3)}\, .
  \end{equation} 
  Then,  with \eqref{A} and \eqref{ED},    there is an explicit value for the entropy $A$, and 
  \vskip 0.1cm 
  \centerline{ 
  $A \doteq 1.62352\ldots, \quad D  \doteq 0.97693\ldots, \quad E \doteq 2.60045\ldots$\, .}

 \paragraph{\bf \em Main costs associated to a truncated expansion.} 
 
  \smallskip
  Each real number of the unit interval  admits an infinite continued fraction expansion derived from the dynamical system, which we call its {\em CLCF} expansion.  When truncated at a finite depth,  its  expansion  becomes finite,  as in \eqref{CF}, and  defines a LFT $h:= h_\a$. This expansion  gives rise to a rational $p/q$, (assumed to be irreducible) which is thus written as $p/q= h_\a(0)$.  
  
  \smallskip 
 On the other hand, the $k$-uple $\a$ defines a matrix $M_\a$   and an integer pair $(P,Q)$, called the continuant pair, defined by $(P,Q)^T := M_\a (0,1)^T$.  The equality $P/Q=p/q$,  holds, but, as  the integers $P$ and $Q$ are  not necessarily coprime, the pair $(P, Q)$  does not coincide  with the pair $(p, q)$. The integer $R(Q):=Q/\gcd(P,Q)$, called the reduced continuant,  is an important parameter  that actually dictates the quality
  of the rational approximation  given by the truncation of the {\em CLCF}.  
  We will see that it also plays a central role in the analysis of the {\em CL} algorithm. 
  As $\gcd(P, Q)$  divides $\left|\det (M_\a)\right|$ that is a power of two, it is itself a power of two. It  then  proves fundamental  to  deal with  dyadic tools. 
  
\smallskip
  The main interesting costs associated with a finite expansion, as in \eqref{CF},   are defined  via the  LFT $h$ and mainly  involve the continuant pair $(P, Q)$  together with the  absolute value of the determinant of the LFT $h$, denoted as $d(h)$.  The next result describes  these costs and provides  alternative expressions.

\begin{proposition} \label{G2}Consider the function $G_2: \Q \rightarrow  {\mathbb R}^+ $ (called the $\gcd$ map)   equal to   $ G_2(y) = \min ( 1, |y|_2^{-2})$, namely  
  \begin{equation} \label{G}     G_2(y) = 1 \qquad \hbox {for $|y|_2  \le   1$}, \qquad G_2(y) =  |y|_2^{-2} \quad   \hbox {for $|y|_2 > 1$}\, . 
  \end{equation}
  The main  costs  associated with the   {\em CLCF} expansion of a rational $h(0)$  
   \vskip 0.15cm 
 \centerline{
$Q, \quad g(P,Q):=\gcd(P, Q),  \quad R(P, Q) ={Q} /{\gcd (P, Q)}, \quad |Q|_2\, .$}
\vskip 0.1cm
 are all expressed  in terms of the  quadruple $(|h'(0)|,   |h'(0)|_2, d(h), G_2[h(0)]$ as
 \vskip 0.1cm 
 \centerline{
 $ Q^{-2}  =   |h'(0)|/d( h), \quad |Q|^{-2}_2 =  d(h) \, |h'(0)|_2 , $}
  \vskip 0.2cm 
  \centerline{$ R^{-2}(Q) =   |h'(0)| \,  | h'(0)|_2 \,  G_2 [h(0)] , \qquad g^2(P, Q) = d(h)\,  |h'(0)|_2 \, G_2[h(0)]
 \,. $}
  \end{proposition}

 \begin{proof} 
 One has (by definition)
 \vskip 0.1cm 
 \centerline{
$ P /Q = h(0), \qquad   Q^{-2} =  {|h'(0)|}/ d(h), \qquad   r(Q)^{-2} ={g^2(P, Q)}/ {Q^2} \, .$}
\vskip 0.1cm
 As $g(P, Q)$ is a power of 2, and using the function $G_2$ defined in  \eqref{G}, one has 
  \vskip 0.15cm
\centerline{
 $ g^{2}(P, Q) =  \min( |P|_2,  |Q|_2)^{-2}  =   |Q|_2^{-2} \min (1, |P/Q|_2^{-2} ) = |Q|_2^{-2}   G_2(P /Q)$.}
  \vskip 0.1cm
  We conclude with the equalities : \ \   
  $|{Q}|^{-2}_2 =   {|h'(0)|_2}/{ |d(h)|_2},  \ \  d(h) \cdot |d(h)|_2 = 1$. 
  \end{proof}

  Any cost $C$ of Proposition  \ref{G2}   admits  an expression of the form
\vskip 0.1cm
\centerline{  
 $ |h'(0)|^t \  |h'(0)|_2^u\  d(h)^v \ G_2[h(0]^z$\, .}
  \vskip 0.1cm
 The   quadruple $(t, u, v, z)$  associated with  the cost $C$ is denoted as $\gamma_C$. Moreover, 
 as  these costs   $C$ are expected to be of  exponential  growth with respect to the depth of the {\em CF},   we   will work  with their logarithms $c= \log C$.  Figure \ref{costs} summarizes  the result.   

 \vspace{-0.3cm}
 \begin{figure}
\begin{center}
{\rm 
\begin{tabular}{|c|c|c|c|l|}
\hline
Cost $C$ &$c= \log C$ &Quadruple   $ \gamma_C$ &  Constant  $M(c)$ &Numerical value of $M(c)$
\cr
\hline
\hline
$d(h)$ &$\sigma $ & $(0,0, 1, 0)$  & $D$ &  $\doteq 0.97693\ldots$
 \cr
$Q^{2}$&$q$ &$(-1,0, 1, 0)$ & $A +D$& $\doteq 2.60045 \ldots$
\cr
$g^2(P, Q)$ & $\varrho$ & $(0, 1, 1, 1)$ &  $B+ D$ & $\doteq 1.26071 \ldots $
\cr 
$ R^{2}(P, Q)$& $r$ & $(-1, -1, 0, -1)$ & $A- B$& $\doteq 1.33973 \ldots$
\cr
$|Q|_2^{-2} $ &$ q_2$ & $(0, 1, 1, 0)$ & $B+D$&  $\doteq 1.26071 \ldots$
\cr

\hline
\end{tabular}
}

\caption {\label {costs}\small{Main costs of interest,  with their quadruple, and the constants which intervene in the analysis  of their mean values. (see Thm 1). }}
\end{center}
\end{figure}

\paragraph*{\bf \em Generating functions. } We  deal with 
 sets of  coprime{\footnote{ This restriction can be easily removed and our analysis extends to the set of all integer pairs.}}  integer pairs 
\vskip 0.1cm 
 \centerline {$ \Omega := \{ (p, q) \mid 0 < p <q, \ \ \gcd (p, q) = 1 \}, \qquad  \Omega_N :=  \Omega \cap \{  ( p,  q) \mid q\le N \} $.}
 \vskip 0.1cm
 The set $\Omega_N$ is endowed with the uniform measure, and  we wish to study the mean values $\E_N[c]$ of parameters $c$  on $\Omega_N$. We focus on  parameters  which describe the execution of the algorithm  and are  ``read'' from  the {\em CF}$(p/q)$  built by the algorithm as in \eqref{CF}. They are defined in Proposition \ref{G2} and depend on the continuant  pair $(P,Q)$;  as already explained, the reduced continuant $R(P, Q)$  plays a fundamental role here.
 
 We deal with analytic combinatorics methodology and  work with (Dirichlet) generating functions  (dgf in short). Here is 
 the plain Dirichlet  generating function 
 \vskip -0.1cm
 \begin{equation} \label {Ss}
  S(s) := \sum_{(p, q)\in \Omega} \frac 1 {q^{2s}}  = \frac {\zeta (2s-1)} {\zeta(2s)}\, . 
  \end{equation}
 There are also two generating functions that are  associated with  a cost $C: \Omega \rightarrow {\mathbb R}^+$ (and its logarithm $c$),  
 namely the bivariate dgf and the cumulative dgf,   \vskip -0.3cm
 \begin{equation} \label{Ssw} 
  S_C(s, w) := \!\! \sum_{(p, q)\in \Omega} \frac {e^{wc(p, q)}} {q^{2s}},\ \  \hat S_C(s) := \!\!\sum_{(p, q)\in \Omega}  \frac {c(p, q)}{q^{2s}} = \frac{\partial}{\partial w} S_C(s, w)\bigg|_{w = 0} \, .
  \end{equation}\vskip -0.1cm
  The expectation $\E_N[c]$ is now expressed as a ratio  which involves the sums  $\Phi_N(S)$, $\Phi_N(\hat S_C)$ of  the  first $N$ coefficients of the  Dirichlet series $S(s)$ and $\hat S_C(s)$,
   \begin{equation} \label{sumcoef} E_N[c] = \Phi_N[\hat S_C] / \Phi_N[S]\,.
   \end{equation} 
   From principles of Analytic Combinatorics, we know  that the dominant singularity of a dgf (here its singularity of largest real part) provides precise information  (via notably its position and its nature)  about the asymptotics of  its coefficients, here closely related to the  mean value $\E_N[c]$ via Eqn \eqref{sumcoef}. Here, in the Dirichlet framework,  this transfer  from  the analytic behaviour  of the dgf  to the asymptotics of its coefficients is provided by Delange's  Tauberian Theorem \cite{De}.  
 
 \smallskip
 We   now describe  an alternative expression of these series, from which  it is possible to  obtain information regarding the dominant singularity, which will be transfered to the asymptotics of coefficients.  

\begin{proposition} \label{prop:gamma}The Dirichlet generating $S(s)$  and its bivariate version $S_C(s, w)$ relative to a cost $ C: \Omega \rightarrow {\mathbb R}$, 
admit  alternative expressions\footnote{We recall that the last exponent is $0$ by convention, and the last LFT is thus $J = h_0$.}
$$S(s)  = S_C(s, 0),  \qquad 
 S_C(s, w) =  \sum_{h \in {\cal H}^\star \cdot J}  e^{w C (h)}\  |h'(0)|^s  \   | h'(0)|^s_2 \   G_2^s \circ  h(0) \,  \, .$$ For any cost $C$  described in Figure \ref{costs},  the  general term of  $S_C(s, w)$ is  of the form 
\vskip 0.15cm
\centerline{$
|h'(0)|^t  \  | h'(0)|^u_2 \    d(h)^v \ G_2^z \circ  h(0) { \, , } $}
\vskip0.1cm
 and    involves a   quadruple  of exponents $(t, u, v, z)$,  denoted as $\gamma_C(s, w)$,   that   
 is expressed  with the quadruple $\gamma_C$  defined in  Figure \ref{costs}  as
   \begin{equation} \label{gamma}
\gamma_C(s, w)= s (1, 1, 0, 1) + w \, \gamma_C\, .
\end{equation}
\end{proposition}
\begin{proof}   By definition, the denominator $q$  equals $R(P, Q)$.  With Figure \ref{costs}, the quadruple relative  to $q^{-2s}$ is then  $s(1, 1, 0, 1)$, whereas the quadruple  relative  to $e^{wc} = C^w$ is  just $w\gamma_C$.
\end{proof}

\smallskip We have thus described the general framework of our paper.  We  now look for an alternative form for the generating functions: in dynamical analysis,  one expresses the  dgf in terms of  the transfer operator of the dynamical system which underlies the algorithm.  Here, it is not possible to obtain such an alternative expressions if we stay in the real ``world''. This is why  we will  add a component to  the {\em CL}  system which allows us to express parameters  with a dyadic flavour.   It will be possible to express  the dgf's in term of a (quasi-inverse) of an (extended) transfer operator, and relate  their dominant singularity to the dominant eigenvalue of this extended transfer operator. 

\smallskip We then obtain our main result, precisely stated in Theorem 1, at the end of the paper:  
we  will  prove that the mean values $E_N [\log C]$   associated with our costs of interest   are  all of order $\Theta(\log N)$, and  satisfy precise asymptotics that involve  three constants $A, B, D$: the constants $D$ and $A$  come from the real word, and  have been previously defined in \eqref{ED} and \eqref{A},  but there arises  a new constant $B$  that comes from the dyadic word.

\section{The extended dynamical system.}

In this section,  we extend the {\em CL} dynamical system, adding a new component to study  the dyadic nature of our costs.  We  then introduce transfer operators, and   express  the  generating functions in terms of the quasi-inverses of the transfer operators.   

\vspace{-0.2cm} 
  \paragraph{\bf \em Extension of the dynamical system. } We will work with  a two-component  dynamical system:  its  first component is the initial {\em CL} system, to which we add a second (new) component which is used to  ``follow'' the evolution of dyadic phenomena during the execution of the first component.   
  
   We consider the set $\und{\cal I} := {\cal I} \times {\mathbb Q}_2$. We define a new shift   $\und T: \und {\cal I}  \rightarrow  \und {\cal I}$ from the characteristics of the old shift $T$ defined in  \eqref{eq:branchescl}.  As  each  branch  $T_a$, or its inverse $h_a$, is a LFT with rational coefficients,  it is well-defined on $\Q$;  it is moreover a  bijection from  $\Q \cup \{\infty\} $ to $\Q \cup \{\infty\} $.  Then, each branch $\und T_a$ of the new shift $\und T$  is  defined via the equality $\und T_a (x,y):=(T_a (x), T_a (y))$  on the fundamental domain $\und{\cal I}_a:= {\cal I}_a \times \Q$,  and the shift $\und T_a$ is a bijection from $\und{\cal I}_a $ to $\und{\cal I} := {\cal I} \times \Q$ whose inverse branch  $\und h_a\colon (x, y) \mapsto  (h_a(x), h_a(y))$ is  a bijection from $\und{\cal I} $ to $\und{\cal I}_a$.

  \vspace{-0.2cm} 
  \paragraph{\bf \em Measures.} We consider the three domains
   \vskip 0.1cm
   \centerline{
   $ {\cal B} := \Q \cap \{ |y|_2 < 1\}, \ {\cal U} :=\Q \cap \{ |y|_2 = 1\},   \ {\cal C} := \Q \cap \{ |y|_2 > 1\}.$}
\vskip 0.1cm 
 There exists a Haar measure $\nu_0$ on $\Q$ which is finite on each compact of $\Q$, and can be normalized with $\nu_0 ({\cal B}) = \nu_0({\cal U}) = 1/3$ (see  \cite{Ko}).  We will deal with the measure $\nu$ with density $G_2$ wrt to $\nu_0$, for which $\nu ({\cal C}) = 1/3$. The measure $\nu[ 2^k {\cal U}]$ equals $(1/3) 2^{-|k|}$ for any $k \in {\mathbb Z}$ and $\nu$ is a probability measure on $\Q$.   On $\und {\cal I}$,  we  deal with the  probability measure  $\rho := \mu \times \nu$ where  $\mu$ is the Lebesgue measure  on ${\cal I}$ and $\nu$ is defined on $\Q$ as  previously.
 
 \smallskip
  For  integrals which  involve   a  Haar measure, there is a  change of variables  formula\footnote{This  general result can be found for instance in Bourbaki \cite{Bou},  chapitre 10, p.36.}. As $\nu_0$ is a Haar measure, and $d \nu = G_2\,  d\nu_0$,   this leads to  the following change of variables  formula, for any   $F \in L^1(\Q, \nu)$, 
 \begin{equation} \label{cofv}\int_{\Q} |h'(y)|_2 \   F(h(y)) \ \left[\frac {G_2 (h(y))}{G_2(y)}\right] \  d \nu(y)=  \int_{\Q}   F(y) \ d \nu(y)  \, . 
 \end{equation}

\vspace{-0.2cm}   
\paragraph{\bf \em Density transformer and transfer operator. }  
 We now consider the   operator  $\und {\bf H}$ 
 defined as  a ``density transformer'' as follows : 
 with  a  function $F \in L^1(\und {\cal I}, \rho)$,      it associates  a  new function  defined by
 $$ \und {\bf H} \, [F ](x, y) :=\sum_{h \in {\cal H} } |h'(x)| \  |h'(y)|_2\  F(  h(x),  h(y)) \  \left[\frac {G_2 (h(y)) }{ G_2(y) }\right]  \, .  $$
  When $F$ is a density  in  $L^1(\und {\cal I}, \rho)$,  then  $\und {\bf H}[F]$ is indeed  the new density on $\und{\cal I}$ after one iteration of the shift $\und T$.   
 This just  follows from the  change of variables formula  \eqref{cofv} applied to   each  inverse branch $h \in {\cal H}$.  
 
 \smallskip
 Proposition \ref{prop:gamma} leads us to  
 a  new operator 
  that depends on  a quadruple $(t, u, v, z)$,
 \begin{equation} \label{Htuvz}  {\bf H}_{t, u, v, z}  [F](x, y) :=   \sum_{h \in\cal H} |h'(x)|^t \   |h'(y)|^u_2 \  d(h)^v  F(  h(x),  h(y)) \ \left[ \frac { G_2 (h(y))}{G_2(y)}\right]^{z}. 
 \end{equation}
We  will  focus on  costs described in Figure \ref{costs}:   we  thus deal with  operators associated with 
 quadruples   $ \gamma_C(s, w)$ defined in Proposition \ref {prop:gamma}, and in particular with the quadruple $(s, s, 0, s)$, and its associated operator 
 $\und{\bf H}_s := {\bf H}_{s, s, 0, s}$. 
  
\vspace{-0.2cm}   
\paragraph{\bf \em Alternative expressions of the Dirichlet generating functions.}  

We start with the expressions  of  Proposition \ref{G2},   consider the three types of dgf defined in \eqref{Ss} and \eqref{Ssw},  use the  equality  $G_2 (0) = 1$, and consider the  operator  $\und {\bf J}_s$ relative to the branch $J$ used in the last step.   
For the plain dgf in \eqref{Ss},  we obtain
 \begin{equation} \label{Ssalt}
  S(s)= \sum_{h \in {\cal H}^\star\!\cdot J}  |h'(0)|^s  \     | h'(0)|^s_2  \   G_2^s \circ  h(0)
  = \und {\bf J}_s \circ  (I- \und {\bf H}_s)^{-1} [ 1] (0, 0)\, ,
  \end{equation}
We now consider 
the  bivariate  dgf's  defined in  \eqref{Ssw}.  For the  depth $K$, one has  
 $$ S_K(s, w) =  e^w \und {\bf J}_s \circ  (I-  e^w \und {\bf H}_s)^{-1} [ 1] (0, 0)\, ;  $$
  For costs $C$ of  Figure \ref{costs},  the  bivariate  dgf   involves the quasi-inverse of  ${\bf H}_{\gamma_C(s,w)}$, 
   $$ S_C(s, w)= 
   {\bf J}_{\gamma_C(s, w)} \circ  (I- {\bf H}_{ \gamma_C (s, w)}) ^{-1} [ 1] (0, 0)\, , $$ 
    except  for $C= |Q|_2^{-2}$,  where  the function $1$ is replaced by the function $G_2^w$.

 \smallskip  
   The dgf $\hat S_C(s)$ defined in \eqref{Ssw}  is obtained with taking the derivative of the bivariate dgf wrt $w$ (at $w = 0$);  it is thus written with  a double\footnote{There is another term which involves only a quasi-inverse. It does not intervene in the analysis.}  quasi inverse  which involves the plain operator $\und{\bf H}_s$,  separated  ``in the middle'' by  the  cumulative operator   $ \und {\bf H}_{s, (C)}$, namely  
 \begin{equation} \label{hatSC} \hat  S_C(s)\asymp   \und{\bf J}_s \circ  (I- \und{\bf H}_s)^{-1}\circ  \und {\bf H}_{s, (C)}  \circ  (I- \und{\bf H}_s) ^{-1}\ [ 1] (0, 0) \, ,
 \end{equation}
$$ \hbox{and the cumulative operator is itself defined by} \quad  
    \und {\bf H}_{s, (C)}:=  \frac{ \partial}{\partial w} {\bf H}_{\gamma_C(s, w)} \Big|_{w = 0}  \, . \hskip 5.2cm $$

\section{Functional Analysis} 
   
 This section deals with  with a delicate context, which mixes the specificities of each world --the real one, and the dyadic one--.    It is devoted to the study of the quasi-inverses $(I-\und {\bf H}_s)^{-1}$ intervening in the expressions  of the generating functions of interest.   We first define an appropriate functional space  
 on which  we prove  the  operators  to act and admit  dominant spectral properties.  This entails  that the quasi-inverse   $(I-\und {\bf H}_s)^{-1}$ admits a  pole at $s= 1$, and we study its residue, which  gives rise to the constants that appear in the expectations of our main costs.

\vspace{-0.2cm}    
\paragraph{\bf \em Functional space. }   The delicate point of the dynamical analysis is the choice of a  good functional space,  that must be a subset of $L^1(\und{\cal I}, \rho)$. Here, we know that, in  the initial {\em CL} system, the transfer operator ${\bf H}_s$ acts in a good way on ${\cal C}^1({\cal I})$. Then, for a function $F$ defined on ${\underline{\mathcal{I}}}$,  the main role will be  played by the family of ``sections'' 
$\tilde F_y: x \mapsto \max (1, \log |y|_2)\,  F(x, y)$ 
which will be asked to belong to ${\cal C}^1({\cal I})$, 
  under the norm $| \cdot |_{1, 1}$,  defined as  $ |\tilde F_y |_{1, 1}  :=  |\tilde F_y |_{0} + |\tilde F_y |_{1}$ with
  
  \vskip 0.1cm
  \centerline{
$|\tilde F_y |_{0} := \sup _{x \in {\cal I} }|\tilde F(x, y)|, \quad    |\tilde F_y |_{1} :=\sup_{x \in {\cal I}}  \left| \frac{\partial}{\partial x}\tilde F(x, y)\right| \, . $}
\vskip 0.2cm 
We  work on the Banach space
\vskip 0.1 cm
\centerline{
$ {\cal F} :=\left\{ F: \und {\cal I} \rightarrow {\mathbb C} \mid   F_y  \in {\cal C}^1 ({\cal I}), \quad  y \mapsto   \tilde F_y \ \ {\rm bounded}  \right\}\,,$}
\vskip 0.1cm
endowed with the norm $ ||F|| := ||F||_0 + ||F||_1$,    with
\begin{equation}
 ||F||_0 :=  \int_{\Q}    |\tilde F_y |_{0}\  d \nu(y), \quad ||F||_1:=  \int_{\Q}  |\tilde F_y  |_{1}  \  d \nu(y) \, .
\end{equation}

 The next Propositions \ref{act}, \ref{dec} and \ref{lambda} will  describe  the behaviour  of the operator ${\bf H}_{t, u, v, z}$ on the functional space ${\cal F}$.  Their proofs  are  quite technical and are  omitted here. 
 
 \smallskip The first result exhibits 
  a subset of quadruples  $(t,u,v,z)$ which contains $(1, 1, 0, 1)$  for which the resulting operator ${\bf H}_{t, u, v, z}$  acts on ${\cal F} $.

\begin{proposition}  \label{act} The following holds:

$(a)$  When the  complex triple $(t, u, v)$   satisfies the constraint $\Re (t-v-|u-1|) >0$, 
  the operator   ${\bf H}_{t, u, v, u}$ acts on ${\cal F}$ and is analytic with respect to the   triple $(t, u, v)$. 

\smallskip
 $(b)$
 The operator $\und {\bf H}_s := {\bf H}_{s, s, 0, s}$ acts on ${\cal F}$  for  $\Re s >1/2$, and  the norm $||\cdot||_0$ of the operator $\und{\bf H}_s$ satisfies
 $||\und{\bf H}_s ||_0<1$ for $\Re s >1$. 
 
 \end{proposition} 

\vspace{-0.2cm}  
\paragraph{\bf \em Dominant spectral properties of the operator.} 
The next result describes some of the main   spectral properties  of the operator on the space ${\cal F}$. Assertion $(a)$  entails that the $k$-th iterate of the operator behaves as a  true $k$-th power of its dominant eigenvalue. Then, as stated in $(c)$,   its quasi-inverse behaves as a true quasi-inverse  which involves its dominant eigenvalue.

\begin{proposition}
\label{dec}
 The following properties  hold  for the operator $ {\bf H}_{t, u, v, u}$,  when the triple $(t, u, v)$   belongs to a neighborhood ${\cal V}$ of $(1, 1, 0)$.

 \begin{itemize}

 \item [$(a)$]  There is a unique   dominant eigenvalue, separated from the remainder of the spectrum by a spectral gap, and  denoted as $\lambda(t, u, v)$, with  a (normalized)  dominant eigenfunction $ \Psi_{t, u, v}$  and  a dominant eigenmeasure $\rho_{t, u, v}$  for  the dual operator.

\smallskip 
 \item[$(b)$] At $(t, u, v, u) = (1, 1, 0, 1)$, the operator  $ {\bf H}_{t, u, v, u}$ coincides with the density transformer $\und{ \bf H}_1$. At $(1, 1, 0)$ the dominant eigenvalue 
$ \lambda(t, u, v)$ equals 1, the function $\Psi_{t, u, v}$ is the  invariant density $\Psi$  and the measure $\rho_{t, u, v}$ equals the measure $\rho$.
 
 \smallskip
 \item [$(c)$] 
The  estimate holds  for any  function $F \in L^1(\und {\cal I}, \rho)$  with  $\rho[F]\not = 0$, 

\vspace {-0.2cm}
 $$\  (I-  {\bf H}_{t, u, v, u})^{-1} [F](x, y)  \sim \frac  {\lambda(t, u, v)}{1- \lambda(t, u, v)} \  \Psi_{t, u, v}(x, y) \,  \rho_{t, u, v}[F] \, .$$
 
 \smallskip
 \item[$(d)$]  For $\Re s = 1, s \not = 1$, the spectral radius of ${\bf H}_{s, s, 0, s}$ is strictly less than 1. 
\end{itemize}
\end{proposition}

The third result  describes the  Taylor expansion   of $\lambda(t, u, v)$ at $(1, 1, 0)$, and makes precise the behaviour of the quasi-inverse  described in $(c)$.  

\begin{proposition}\label{lambda}
 The  Taylor expansion of  the eigenvalue $\lambda(t, u, v)$  at $(1, 1, 0)$, written as   
$\ \lambda(t, u, v) \sim  1 -  A(t-1)   +  B(u-1) +Dv  $, \ involves  the  constants 
\vskip 0.1cm 
\centerline{
$ A = - {\partial\lambda}/{\partial t}  (1, 1, 0) , \ \   B =   {\partial \lambda}/{\partial u}  (1, 1, 0), \ \  D =   {\partial \lambda}/{\partial v} (1, 1, 0)\,   $}
\vskip 0.2cm

 $(a)$
 The constants $A$ and $D$ already appear in the context of the plain dynamical system, and are precisely described in \eqref{ED} and \eqref{A}. In particular $A-D$   is equal to the integral $E := \E_\Psi [2 |\log x|] $; 

\smallskip
$(b)$ 
The  constant $B$ is defined with the extension of the dynamical system and its invariant density $\Psi = \Psi_{1, 1, 0}$.  The constant $B+D$   is  equal to the  dyadic analog  $E_2$  of the integral $E$, namely,  $ B+ D = E_2:=    \E_\Psi [2  \log |y|_2 ] $; 

\smallskip
 $(c)$ The constant $A-B $
 is  the entropy of the extended dynamical system. \end{proposition}

\vspace{-0.3cm}
\paragraph{\bf \em Final result for the analysis of the {\em CL} algorithm.} 
 We then obtain our final result:  
 \begin{theorem}  \label{thm1}
 The mean values $E_N [c]$ for $c \in \{ K,  \sigma, q, \varrho, r, q_2\}$   on the set $\Omega_N$ are  all of order $\Theta(\log N)$ and  admit  the precise  following estimates,  
 $$ E_N[K]  \sim  \frac 2{H} \log N ,  \qquad  \E_N[c] \sim  M(c) \cdot \E_N[K],  \quad \hbox{for} \ \  c \in \{\sigma, q, \varrho, r, q_2\}\, .$$ The constant $H$ is the entropy of the extended system.   The  constants $H$ and $M(c)$ are expressed  with   a scalar product that involves
 the  gradient $\nabla \lambda $ of the dominant eigenvalue at $(1, 1, 0)$ and the  beginning $\hat \gamma_C$ of the  quadruple  $ \gamma_C$ associated with the cost $c$. More precisely
 \vskip 0.1cm 
 \centerline{
 $H = -\langle \nabla \lambda, (1, 1, 0)\rangle , \qquad M(c) =  \langle \nabla \lambda,  \hat \gamma_C\rangle \, . $}
 \vskip 0.1cm 
The constants $M(c)$  are  exhibited in Figure \ref{costs}.
\end{theorem} 
\vspace{-0.2cm}

\begin{proof}    Now, the Tauberian Theorem comes into play, relating the behaviour of  a Dirichlet series $F(s)$ near its dominant singularity  with asymptotics for the sum $\Phi_N(F)$ of its first $N$ coefficients.  Delange's Tauberian Theorem  is stated as follows (see \cite{De}):  
 
  \smallskip
    {\em Consider for $\sigma >0$  a Dirichlet series $F(s) := \sum_{n \ge 1} a_n n^{-2s}$ with non negative coefficients  which converges for $\Re s>\sigma$. Assume moreover: 
    \smallskip
   
   \hskip 0.5cm  $(i)$ $F(s)$ is analytic on $\{\Re s, s \not = \sigma\}$,  
  
  \hskip 0.5cm  $(ii)$ near $\sigma$,  $F(s)$ satisfies  $F(s) \sim A(s)(s-\sigma)^{-(k+1)}$  for some integer $k \ge 0$ .

\smallskip   
Then, as $N \to \infty$,   the sum $\Phi_N(F)$ of its first $N$ coefficients  satisfies 
\vskip0.1 cm 
\centerline{
$\Phi_N(F):= \sum_{n \le N} a_n  \sim  2^k \, A(\sigma) \,  [\sigma \Gamma(k+1)]^{-1}  \,  N^{2\sigma}\,  \log^k N\ .$}
\vskip 0.1cm}

We  now show that the two dgf's $S(s)$ and $\hat S_C(s)$ satisfy the hypotheses of the Tauberian theorem.
  The two expressions obtained in \eqref{Ssalt} and \eqref{hatSC}  involve quasi-inverses  $(I-\und {\bf H}_s)^{-1}$,  a  simple one in  \eqref{Ssalt}, a  double  one in   \eqref{hatSC}. 
  
 \smallskip 
     First, Propositions \ref{act}$(b)$ and  \ref{dec}$(d)$ prove that  such  quasi-inverses  are  analytic on  $\Re s \ge 1, s \not = 1$.  Then Proposition \ref{dec}$(c)$,  together with Eqn \eqref{hatSC},  shows that $S(s)$ and $\hat S_C(s) $ have a pole at $s = 1$, of order 1 for $S(s)$, of order 2 for $\hat S_C(s)$.
     
     \smallskip
  We now evaluate the dominant constants:  first, the estimate  holds, 
  \vskip 0.1 cm
  \centerline{
$1- \lambda(s, s, 0) \sim (A-B) (s -1) = H(s-1)$, with $H = - \langle \nabla \lambda, (1, 1, 0)\rangle $.}

\vskip 0.1 cm 
Second,    with Proposition \ref{dec}$(c)$,    the  dgf's $S(s)$ and  $\hat S_C(s) $ admit  the  following estimates  which  both involve the constant $a=\und {\bf J}[\Psi](0, 0)$, namely, 

 \vspace{-0.15cm}
 $$  S(s) \sim  \frac {a} {H (s -1)}, \quad  \hat  S_C(s) \sim  \frac {a} {H^2 (s -1)^2}\  \rho \left[\und {\bf H}_{1, (C)}[\Psi] \right] \, .$$
 
 \vspace{0.1cm}
   We now  explain the occurence of the constant $M(c)$:  we use  the definition of the triple $\hat \gamma_C(s, w)$, the definition of the cumulative operator   $\und {\bf H}_{1, (C)}$ as the derivative of the bivariate operator ${\bf H}_{\gamma_C(s, w)}$ at $(s, w) = (1, 0)$,  and the fact that ${\bf H}_{1, 1, 0, 1} = \und {\bf H}_1$ is the density transformer. This entails  
 the sequence of equalities, 
 
 \vspace{-0.15cm}
$$  \rho\left[ \und {\bf H}_{1, (C)} [\Psi] \right]   =  \frac{\partial}{\partial w} \lambda(\hat  \gamma_C(1, w))\Big|_{w = 0} = \langle \nabla \lambda , \hat \gamma_C \rangle  = M(c)\, .$$
   \end{proof}

\vspace{-0.2cm}
\paragraph{\bf \em  About the constant $B$.} 
The  invariant  density  $\Psi$ --more precisely   the function 
$\hat \Psi := \Psi\cdot   G_2$ --
satisfies a functional equation  of the same type as the invariant function $\psi$,  (described in Eqn \eqref{funceqn}),  namely, 
$$  \hat \Psi (x, y) = 
 \left(\frac 1 {1+x}\right)^2 \ \left|\frac 1 {1+y}\right|_2^2 \ \sum_{a \ge 0}  \hat  \Psi \left(\frac {2^{-a}}{1+x},  \frac {2^{-a}}{1+y}\right)\, .$$
  Comparing to Eqn \eqref{funceqn}, we ``lose'' the factor $2^{-a}$ in the sum,   and so we have not succeed in finding an explicit formula for $\Psi$. We do not   know how to evaluate the integral $E_2$ defined in Proposition \ref{lambda}$(b)$. 
  
  \smallskip
 However, we conjecture\footnote{This will be explained in the long paper.}  the equality $D-B= \log 2$,  from experiments of the same type as those described in Figure \ref{exe}. This would entail an explicit value for the entropy 
  of the extended system, 
 $$ \frac 1 {2 \log 2\! - \!\log 3 } \left[  \frac {\pi^2}{6} +  2 \sum_{k \ge 1}  \frac {(-1)^k}  {k^2\,  2^k} - (\log 2) (3 \log 3\!- \!4 \log 2) \right]  \doteq 1.33973\ldots
 $$

\paragraph{\bf \em Conclusions and Extensions.}
   We have studied the Continued Logarithm Algorithm 
and    analyzed  in particular the number of pseudo divisions, and the total number of shifts. It would be nice to obtain an explicit expression of the invariant density, that should entail a proven expression of the entropy of the dynamical system.    It is  also surely possible to  analyze  the bit complexity of the algorithm,   notably in the case when one eliminates the rightmost zeroes   when  are shared by the two $q_i's$ (as suggested by Shallit). Such a version of this algorithm  may have a competitive bit complexity that  merits  a further study. 
   
 \smallskip  
 There exist two other gcd algorithm that are based on binary shifts, all involving a dyadic point of view: the Binary Algorithm, and   ``the Tortoise and the Hare'' algorithm, already analyzed in   \cite{Va1} and  \cite{DaMaVa}; however, the role of the binary shifts is different in each case.   The strategy of the present algorithm is led by the most significant bits, whereas  the strategy of the ``Tortoise and the Hare'' is led by the least significant bits.  The Binary algorithm  adopts  a mixed strategy, as it performs both right-shifts  and subtractions. We have the project to unify the analysis of  these three algorithms, and better understand the role of the dyadic component in each case.


\begin{thebibliography}{99} 


\bibitem{bor}
Jonathan M. Borwein, Kevin G. Hare and Jason G. Lynch. \newblock Generalized Continued Logarithms and Related Continued Fractions, \newblock Journal of Integer Sequences, vol. 20 , 2017.

\bibitem{Bou} Nicolas Bourbaki, Vari\'et\'es diff\'erentielles et analytiques.  Springer 2007

\bibitem{chan1} Hei-Chi Chan. The Asymptotic Growth Rate of Random Fibonacci Type Sequences.  {Fibonacci Quarterly, vol. 43, no. 3, pp. 243--255, 2005.}

\bibitem{DaMaVa} Benoit Daireaux, V\'eronique Maume-Deschamps, Brigitte Vall\'ee.  The Lyapounov Tortoise and the Dyadic hare. Proceedings of AofA'05, DMTCS,  pp. 71--94, 2005.

\bibitem {De} {Hubert Delange.} 
\newblock {G\'en\'eralisation du Th\'eor\`eme d'Ikehara,} 
\newblock {Ann. Sc. ENS} 71,   pp  213--242, 1954 

\bibitem{gosper} Bill Gosper. \newblock Continued fraction arithmetic, \newblock {{ Unpublished manuscript}}, ca. 1978. 

\bibitem{Ko} Neal Koblitz. \newblock  $p$-adic Numbers, $p$-adic analysis  and Zeta functions. 2nd edition,  Springer Verlag, 1984

 \bibitem{shallit} Jeffrey Shallit. \newblock  Length of the continued logarithm algorithm on rational inputs. \newblock https://arxiv.org/abs/1606.03881v2  arXiv:1606.03881v2, 2016.

\bibitem{Va1}  Brigitte Vall\'ee.  Dynamics of the Binary Euclidean Algorithm: Functional analysis and
operators,  Algorithmica, vol. 22, no. 4, pp. 660--685, 1998.

\bibitem{Va2}  Brigitte Vall\'ee. Euclidean Dynamics.  Discrete and Continuous Dynamical Systems, vol. 15, no. 1, pp. 281--352, 2006. 

\end{thebibliography}
\end{document}